\documentclass[english,reprint, longbibliography, superscriptaddress, breaklinks=true, showkeys, showpacs=false, nofootinbib]{revtex4}
\usepackage[T1]{fontenc}
\usepackage[latin9]{inputenc}
\usepackage{color}
\usepackage{babel}
\usepackage{amsmath}
\usepackage{amsthm}
\usepackage{amsfonts}
\usepackage{amssymb}
\usepackage{graphicx}
\usepackage{subfig}
\usepackage{physics}
\newtheorem{mydef}{Definition}
\newtheorem{teo}{Theorem}

\usepackage[unicode=true,pdfusetitle,
 bookmarks=true,bookmarksnumbered=false,bookmarksopen=false,
 breaklinks=true,pdfborder={0 0 0},backref=false,colorlinks=true]
 {hyperref} 
\usepackage{breakurl}

\makeatletter
\@ifundefined{textcolor}{}
{%
 \definecolor{BLACK}{gray}{0}
 \definecolor{WHITE}{gray}{1}
 \definecolor{RED}{rgb}{1,0,0}
 \definecolor{GREEN}{rgb}{0,1,0}
 \definecolor{BLUE}{rgb}{0,0,1}
 \definecolor{CYAN}{cmyk}{1,0,0,0}
 \definecolor{MAGENTA}{cmyk}{0,1,0,0}
 \definecolor{YELLOW}{cmyk}{0,0,1,0}
}

\pdfoutput=1
\hypersetup{colorlinks=true,citecolor=blue,linkcolor=cyan,urlcolor=blue,filecolor= green, breaklinks=true}
\usepackage{url}
\usepackage{breakurl}

\makeatother

\begin{document}

\title{Monogamy and trade-off relations for correlated quantum coherence}

\author{Marcos L. W. Basso}
\email{marcoslwbasso@mail.ufsm.br}
\address{Departamento de F\'isica, Centro de Ci\^encias Naturais e Exatas, Universidade Federal de Santa Maria, Avenida Roraima 1000, Santa Maria, RS, 97105-900, Brazil}

\author{Jonas Maziero}
\email{jonas.maziero@ufsm.br}
\address{Departamento de F\'isica, Centro de Ci\^encias Naturais e Exatas, Universidade Federal de Santa Maria, Avenida Roraima 1000, Santa Maria, RS, 97105-900, Brazil}

\selectlanguage{english}%

\begin{abstract}
One of the fundamental differences between classical and quantum mechanics is in the ways correlations can be distributed among the many parties that compose a system. While classical correlations can be shared among many subsystems, in general quantum correlations cannot be freely shared. This unique property is known as monogamy of quantum correlations. In this work, we study the monogamy properties of the correlated coherence for the $l_1$-norm and relative entropy measures of coherence. For the $l_1$-norm the correlated coherence is monogamous for a particular class of quantum states. For the relative entropy of coherence, and using maximally mixed state as the reference incoherent state, we show that the correlated coherence is monogamous for tripartite pure quantum systems.
\end{abstract}


\maketitle

\section{Introduction}
The wave-particle duality is one of the most fascinating aspects of quantum mechanics, and captures the role of quantum coherence in quantum systems \cite{Bohr}. Another intriguing aspect of quantum mechanics is quantum entanglement \cite{Schrodinger}. One of the reasons is that quantum entanglement, and, in general, quantum correlations, cannot be shared freely among many parties. This property is known as monogamy of quantum correlations \cite{Wooters}.  More recently, it became known that quantum coherence in a composite system can be contained either locally or in the correlations between the subsystems. The portion of quantum coherence contained within correlations can be viewed as a kind of quantum correlation, called correlated coherence \cite{Tan}. Thus, a natural question that arises, in the context of monogamy relations for quantum correlations, is whether correlated coherence also have this property. This is the question we try to answer in this article.

In quantum mechanics, the state of a physical system can be described as a vector in a Hilbert space $\mathcal{H}$. One of the features of such vector space is that any linear combination of vectors also belongs to $\mathcal{H}$, thus allowing superposition of states \cite{Ziman}. This superposition of states is crucial to explain interference patterns in multiple-slit experiments, that otherwise can't be explained by classical physics. One special kind of quantum superposition is quantum coherence, which became an important physical resource in quantum information and quantum computation \cite{nielsen}. More recently, it was shown that quantum coherence is a natural generalization of visibility for quantifying the wave aspect of a quanton in multiple-slit experiments \cite{Bera, prillwitz, Bagan, Tabish, Mishra, basso}. It also has an important role in several research fields, such as quantum biology \cite{Lloyd, Plenio} and quantum metrology \cite{Maccone}. An important step towards the quantification of quantum coherence was given by Baumgratz et al. \cite{Baumgratz}. They established reasonable conditions that a measure of coherence must satisfy to be considered a bona fide measure: Nonnegativity, monotonicity under incoherent completely positive and trace preserving maps (ICPTP), monotonicity under selective incoherent operations on average, and convexity under mixing of states. In the same work, they showed that the $l_1$-norm and the relative entropy of coherence are bone fide measures of coherence, meanwhile the Hilbert-Schmidt (or $l_2$-norm) coherence is not a coherence monotone, i.e., it is not monotone under ICPTP. Later, an equivalent and rigorous framework for quantifying coherence was given in \cite{Yu}.

Another fundamental difference between classical and quantum mechanics is in the ways correlations can be shared among the many parts that compose a system. While classical correlations can be freely shared among many subsystems, in general quantum correlations cannot be freely shared. For instance, if we consider a pair of perfectly correlated classical bits $A$ and $B$, it's possible to maximally correlate the bit $A$ also with a third bit $C$. Mathematically, if e.g. $\rho_{AB} = \frac{1}{2}( \ketbra{0,0} + \ketbra{1,1})$ describes the maximally classically-correlated state of subsystems $A$ and $B$, it's possible to construct a joint quantum state of three subsystems $\rho_{ABC} = \frac{1}{2}( \ketbra{0,0,0} + \ketbra{1,1,1})$ such that $A$ is also maximally classically-correlated with $C$, once the reduced states \cite{pTr} are equal $\rho_{AC}=\rho_{AB}$. In contrast, for two perfectly quantum-correlated (maximally entangled) qubits $A$ and $B$, it is not possible for $A$ to be maximally quantum-correlated with a third qubit $C$. Mathematically, if e.g. $\ket{\Psi}_{A,B} = \frac{1}{\sqrt{2}}(\ket{0,0}_{A,B} + \ket{1,1}_{A,B})$ describes the maximally quantum-correlated state of $A$ and $B$, there is
no tripartite quantum state $\ket{\Psi}_{A,B,C}$ such that $\rho_{AB} = \rho_{AC}$. Hence, if a pair of qubits $A$ and $B$ are maximally entangled, then the system A (or B) cannot be maximally entangled to a third system C \cite{Wooters, Winter}. Actually, the more a qubit is entangled with another qubit, the less it can be entangled with a third one. This fact can be exemplified with the state $\ket{\Psi}_{A,B,C} = x(\ket{0,0}_{A,B} + \ket{1,1}_{A,B}) \ket{j}_C + \sqrt{1 - x^2}(\ket{0,0}_{A,C} + \ket{1,1}_{A,C}) \ket{j}_B$, where $j = 0,1$ and $x \in [0,1]$. One can see that for $x = 1$ the subsystems $A$ and $B$ are perfectly correlated and $A$ and $C$ are uncorrelated, whereas for $x = 0$ the parties $A$ and $C$ are perfectly correlated but $A$ and $B$ are uncorrelated. This indicates that there is a limitation in the distribution of entanglement \cite{Khan}. This unique property, known as entanglement monogamy, has received a lot of attention by researchers \cite{Osborne, Giorgi, Pati, Jin, Guo}. Mathematically, for a tripartite quantum system described by the density matrix $\rho_{A,B,C}$, the monogamy of an arbitrary quantum correlation measure $Q$ is expressed by:
\begin{equation}
    Q(\rho_{A|BC}) \ge Q(\rho_{AB}) + Q(\rho_{AC}), \label{eq:mon}
\end{equation}
where $Q(\rho_{A|BC})$ denotes the quantum correlation $Q$ between $A$ and $BC$ taken as a unit \cite{Fei}. For instance, in the example of the classical bits given above, if the correlation measure is the classical mutual information \cite{nielsen}, we have $I(A:B) + I(A:C) = 2 > I(A:BC) = 1$. In contrast, it was shown in Ref. \cite{Wooters} that the inequality in Eq. (\ref{eq:mon}) is always satisfied for the squared concurrence, an quantum entanglement measure, applied to three qubit states.

However, it is known that entanglement is not the only quantum correlation existing in multipartite quantum systems \cite{Bromley, Adesso, Xi, Xiz, Yao}. For example, quantum discord is a type of quantum correlation that describes the incapacity of a local observer to obtain information about a subsystem without perturbing it \cite{Zurek}. In \cite{Giorgi}, the authors showed that quantum discord and entanglement of formation obey the same monogamy relationship for tripartite pure cases. But, it's known that, in general, quantum discord is not monogamous, except when the quantum states satisfy certain properties \cite{Pati}.
Meanwhile, the monogamy of coherence was first studied in \cite{Segar}, and later addressed in \cite{Zhou}. In addition, Yu \cite{Chang} studied its polygamy. More recently, it became known that quantum coherence in a composite system can be contained either locally or in the correlations between the subsystems \cite{Yao}. The portion of quantum coherence contained within correlations can be viewed as a kind of quantum correlation, called correlated coherence \cite{Tan}. Thus, a natural question that arises, in the context of monogamy relations for quantum correlation, is whether correlated coherence fulfills a monogamy relation like (\ref{eq:mon}). In this article we show that, for the $l_1$-norm, the correlated coherence is monogamous for a given class of quantum states, and we conjecture that is monogamous, at least, for tripartite pure quantum states, what is noteworthy once the $l_1$-norm correlated coherence doesn't vanish for uncorrelated separable states.  Also, for the relative entropy of coherence, and using a maximally mixed state as the reference incoherent state, we show that the correlated coherence is monogamous for tripartite pure quantum systems. Another interesting finding is that the relative entropy of correlated coherence is equal to the quantum mutual information when one uses maximally mixed state as the reference incoherent state. Finally, we also establish some trade-off relations between tripartite and bipartite quantum systems using correlated coherence. It's worth mentioning that trade-off relations between coherence and entanglement were studied in \cite{Song}, and between coherence and correlation were addressed in \cite{Zhang}.\\

We organized the remainder of this article in the following manner. In Sec. \ref{sec:pre}, we give the definition of correlated coherence and discuss some of its properties. The main results of this article are reported in Sec. \ref{sec:tra}, where we show that, for the $l_1$-norm, the correlated coherence is monogamous for a given class of quantum states. Meanwhile, for the relative entropy of coherence, and using maximally mixed state as the reference incoherent state, we show that the correlated coherence is monogamous for tripartite pure quantum systems. Furthermore, for a bipartition of a quantum system, the relative entropy of correlated coherence is equal to the mutual quantum information. We also present some trade-off relations between tripartite and bipartite quantum systems using correlated coherence. Lastly, in Sec. \ref{sec:con} we give our conclusions.

\section{Preliminaries}
\label{sec:pre}
In this section, we give the definition of correlated coherence and explore some of its properties. 
\begin{mydef}
For some coherence measure $C$, the correlated coherence of a multipartite quantum system described by the density operator  $\rho_{A_1,...,A_n}$ in the composite Hilbert space $\mathcal{H}_{A_1} \otimes ... \otimes \mathcal{H}_{A_n}$ is defined as
\begin{equation}
    C^c(\rho_{A_1...A_n}) := C(\rho_{A_1...A_n}) - \sum_{i = 1}^n C(\rho_{A_i}), 
\end{equation}
where $\rho_{A_i}$ is the reduced density matrix of the subsystem $A_i$, with $i = 1,...,n$.
\end{mydef}
In this article, we won't define the correlated coherence as $ C^c(\rho_{A_1...A_n}) = C(\rho_{A_1...A_n}) - C(\bigotimes_{i = 1}^n \rho_{A_i})$,
because not every coherence measure satisfies the property $C(\bigotimes_{i = 1}^n \rho_{A_i}) = \sum_{i = 1}^n C(\rho_{A_i})$, as we'll show in this section.
Now, let $\{\ket{i_m}_{A_m} \}_{i_m = 0}^{d_{A_m} - 1}$ be an orthonormal local basis of the Hilbert space $\mathcal{H}_{A_m}$, with $m = 1,...,n$. Then 
\begin{align}
    \rho_{A_1, ..., A_n} = \sum_{i_1,...,i_n} \sum_{j_1,...,j_n} \rho_{i_1 ... i_n,j_1...j_n}\ket{i_1,...,i_n}_{A_1,...,A_n}\bra{j_1,...,j_n},
\end{align}
and the state of the subsystem $A_m$, which is obtained by tracing over the other subsystems, is given by
\begin{align}
    \rho_{A_m} = \sum_{i_m,j_m}\rho_{i_m,j_m}^{A_m}\ket{i_m}_{A_m}\bra{j_m} = \sum_{i_m,j_m}\sum_{i_{\alpha}, \forall \alpha \neq m}\rho_{i_1 ...,i_m,..., i_n, i_1 ...,j_m,..., i_n}\ket{i_m}_{A_m}\bra{j_m},
\end{align}
where $\sum_{i_{\alpha}, \forall \alpha \neq m}$ means summation for all $i_{\alpha}$ such that $\alpha \neq m$ for a given $m$. Throughout this article, we'll use the following coherence measures:
\begin{mydef}
The $l_1$-norm measure of quantum coherence is defined as
\begin{align}
        C_{l_{1}}(\rho) := \min_{\iota \in I}||\rho-\iota||_{l_{1}} = \sum_{j\ne k}|\rho_{jk}|,
    \end{align}
where $\norm{A}_{l_1} = \sum_{j,k}\abs{A_{jk}}$, for $A \in \mathbb{C}^{n \times n}$, and $I$ is the set of all incoherent states. Meanwhile, the relative entropy of quantum coherence is defined as
\begin{equation}
    C_{re}(\rho)  := \min_{\iota \in I} S(\rho||\iota) = S(\rho_{diag}) - S(\rho),
\end{equation}
where $S(\rho||\iota) = \Tr(\rho \ln (\rho - \iota) )$ is the relative entropy. In both cases, the closest incoherent state is given by $\iota = \rho_{diag} = \sum_j \rho_{jj} \ketbra{j}$.
\end{mydef}

Also, we already omitted the upper limits of the summations for convenience. The following theorem holds for any multipartite quantum system:
\begin{teo}
Let $\rho_{A_1,...,A_n}$ be a multipartite quantum state and let $\{\ket{i_m}_{A_m} \}_{i_m = 0}^{d_{A_m} - 1}$ be an orthonormal local basis for the Hilbert space $\mathcal{H}_{A_m}$, with $m = 1,...,n$. Then $C^c_{l_1}(\rho_{A_1...A_n}) \ge 0$, where $C^c_{l_1}$ is the correlated coherence using the $l_1$-norm as coherence measure.
\end{teo}
\begin{proof}
By definition
\begin{align}
    C^c_{l_1}(\rho_{A_1...A_n})& := C_{l_1}(\rho_{A_1...A_n}) - \sum_{m = 1}^n C_{l_1}(\rho_{A_m})\\
    & := \sum_{(i_1,...,i_n) \neq (j_1,...,j_n)} \abs{\rho_{i_1 ... i_n,j_1...j_n}} -  \sum_{m = 1}^n \sum_{i_m \neq j_m}\abs{\sum_{i_{\alpha}, \forall \alpha \neq m}\rho_{i_1 ...,i_m,..., i_n, i_1 ...,j_m,..., i_n}}, 
\end{align}
where
\begin{align}
    \sum_{(i_1,...,i_n) \neq (j_1,...,j_n)} \equiv
    \sum_{\overset{i_1 \neq j_1}{\overset{i_2 = j_2}{\overset{\vdots}{i_n = j_n}}}} + \sum_{\overset{i_1 = j_1}{\overset{i_2 \neq j_2}{\overset{\vdots}{i_n = j_n}}}} + ... + \sum_{\overset{i_1 = j_1}{\overset{i_2 = j_2}{\overset{\vdots}{i_n \neq j_n}}}} + \sum_{\overset{i_1 \neq j_1}{\overset{i_2 \neq j_2}{\overset{\vdots}{i_n = j_n}}}} + ... + \sum_{\overset{i_1 \neq j_1}{\overset{i_2 = j_2}{\overset{\vdots}{i_n \neq j_n}}}} + ... + \sum_{\overset{i_1 \neq j_1}{\overset{i_2 \neq j_2}{\overset{\vdots}{i_n \neq j_n}}}}.
\end{align}
Using the fact that $\abs{\sum_i z_i} \le \sum_i \abs{z_i}$ $\forall \ z_i \in \mathbb{C}$, we have
\begin{align}
    C^c_{l_1}(\rho_{A_1...A_n})&  \ge \sum_{(i_1,...,i_n) \neq (j_1,...,j_n)} \abs{\rho_{i_1 ... i_n,j_1...j_n}} -  \sum_{m = 1}^n \sum_{i_m \neq j_m}\sum_{i_{\alpha}, \forall \alpha \neq m}\abs{\rho_{i_1 ...,i_m,..., i_n, i_1 ...,j_m,..., i_n}}\\
    & = \Big( \sum_{\overset{i_1 \neq j_1}{\overset{i_2 \neq j_2}{\overset{\vdots}{i_n = j_n}}}} + ... + \sum_{\overset{i_1 \neq j_1}{\overset{i_2 = j_2}{\overset{\vdots}{i_n \neq j_n}}}} + ... + \sum_{\overset{i_1 \neq j_1}{\overset{i_2 \neq j_2}{\overset{\vdots}{i_n \neq j_n}}}}\Big)\abs{\rho_{i_1 ... i_n,j_1...j_n}}\\
    & \ge 0, 
\end{align}
once we have the sum of non negative real numbers.
\end{proof}   

It's worth pointing out that the theorem stated above was first proved in \cite{Tan} for the bipartite case, and implicitly proved in \cite{Jiang} for the three-qubit case, since they proved the following result: $C_{l_1}(\rho_{ABC}) \ge C_{l_1}(\rho_{A}) + C_{l_1}(\rho_{B}) + C_{l_1}(\rho_{C})$.  

Now, let's restrict ourselves to a bipartite quantum system. Then, for a separable uncorrelated quantum state $    \rho_{A,B} = \rho_A \otimes \rho_B = \sum_{i,j,k,l} \rho_{ik}^A \rho_{jl}^B \ket{i,j}_A\bra{k,l}$, the correlated coherence is not necessarily zero:
\begin{equation}
C_{l_1}^c(\rho_{A}\otimes \rho_B) = \sum_{(i,j) \neq (k,l)}\abs{\rho_{ik}^A} \abs{\rho_{jl}^B} - \sum_{i \neq k}\abs{\rho_{ik}^A} - \sum_{j \neq l}\abs{\rho_{jl}^A} = \sum_{\overset{i \neq k}{j \neq l}} \abs{\rho_{ik}^A} \abs{\rho_{jl}^B}= C_{l_1}(\rho_{A})C_{l_1}(\rho_B)\ge 0,
\end{equation}
which implies that $C_{l_1}(\rho_{A}\otimes \rho_B) \neq C_{l_1}(\rho_A) + C_{l_1}(\rho_B)$, as mentioned before (for a related discussion, see \cite{maziero}). Analogously, for a separable state $    \rho_{AB} = \sum_{\alpha} p_{\alpha} \rho_{\alpha}^A \otimes \rho_{\alpha}^B = \sum_{\alpha} \sum_{i,j,k,l} p_{\alpha} \rho_{{\alpha},ik}^A \rho_{{\alpha},jl}^B \ket{i,j}_A\bra{k,l}$, we have
\begin{equation}
    C_{l_1}^c(\rho_{AB}) = \sum_{\overset{i \neq k}{j \neq l}} \sum_{\alpha}p_{\alpha}\abs{\rho_{\alpha,ik}^A} \abs{\rho_{\alpha, jl}^B} \ge 0. 
\end{equation}

However, using the relative entropy of coherence as  the measure of correlated coherence \cite{Kraft}, since $S(\bigotimes_m \rho_{A_m}) = \sum_m S(\rho_{A_m})$, we have $C_{re}(\bigotimes_m \rho_{A_m}) = \sum_m C_{re}(\rho_{A_m})$, and therefore the relative entropy of correlated coherence satisfies $C_{re}^c(\rho_{A_1,...,A_n}) = 0$ for $\rho_{A_1,...,A_n} = \bigotimes_m \rho_{A_m}$. More generally, it is worth mentioning that R\'enyi's entropy is also additive \cite{Shao, Zhu}, hence correlated coherence measures based on R\'enyi's entropy must satisfy such relation.

\section{Monogamy and trade-off relations for correlated coherence}
\label{sec:tra}
\subsection{$l_1$-norm correlated coherence}
In \cite{Jiang}, the authors proved that the conjecture $C(\rho_{ABC}) \ge C(\rho_{AB}) + C(\rho_{AC})$, made by Yao et al. \cite{Yao}, for the $l_1$-norm is invalid.  They considered a counter example using the following the quantum state: $\ket{\Psi} = a_{000}\ket{0,0,0}_{A,B,C} + a_{100}\ket{1,0,0}_{A,B,C}$. However, it is interesting to note that for the correlated coherence this type of trade-off relation holds for the quantum state mentioned above. Since $C^c_{l_1}(\rho_{ABC}) \ge C^c_{l_1}(\rho_{AB}) + C^c_{l_1}(\rho_{AC})$ implies $C_{l_1}(\rho_{ABC}) + C_{l_1}(\rho_A) \ge C_{l_1}(\rho_{AB}) + C_{l_1}(\rho_{AC})$, and $C_{l_1}(\rho_{ABC}) = C_{l_1}(\rho_{A}) = C_{l_1}(\rho_{AB}) = C_{l_1}(\rho_{AC}) = 2\abs{a_{000}a^*_{100}}$. More generally, we have the following theorem:
\begin{teo}
 Let $\rho_{ABC}$ be a tripartite quantum state and let $\{\ket{i}_{A} \}_{i = 0}^{d_{A} - 1}$,$\{\ket{j}_{B} \}_{j = 0}^{d_{B} - 1}$,$\{\ket{k}_{C} \}_{k = 0}^{d_{C} - 1}$ be a orthonormal local basis of  $\mathcal{H}_{A}$, $\mathcal{H}_{B}$, and $\mathcal{H}_{C}$, respectively. If the reduced quantum system $\rho_A = \Tr_{B,C}{\rho_{A,B,C}}$ satisfies
\begin{equation}
    C_{l_1}(\rho_A) = \sum_{i \neq l}\abs{\sum_{j,k}\rho_{ijk,ljk}} =  \sum_{i \neq l}\sum_{j,k}\abs{\rho_{ijk,ljk}}, 
\end{equation}
then
\begin{equation}
    C^c_{l_1}(\rho_{ABC}) \ge C^c_{l_1}(\rho_{AB}) + C^c_{l_1}(\rho_{AC}).
\end{equation}
\end{teo}
\begin{proof}
\begin{align}
    C^c_{l_1}(\rho_{ABC}) & - C^c_{l_1}(\rho_{AB}) - C^c_{l_1}(\rho_{AC}) = C_{l_1}(\rho_{ABC}) + C_{l_1}(\rho_A) - C_{l_1}(\rho_{AB}) - C_{l_1}(\rho_{AC})\\
    & = \sum_{(i,j,k) \neq (l,m,n)} \abs{\rho_{ijk,lmn}} + \sum_{i \neq l}\abs{\sum_{j,k}\rho_{ijk,ljk}} - \sum_{(i,j) \neq (l,m)} \abs{\sum_k \rho_{ijk,lmk}} - \sum_{(i,k) \neq (l,n)} \abs{\sum_j \rho_{ijk,ljn}} \\
    &\ge \Big(\sum_{\overset{i \neq l}{\overset{j \neq m}{k \neq n}}} + \sum_{{\overset{i = l}{\overset{j \neq m}{k \neq n}}}} - \sum_{\overset{i \neq l}{\overset{j = m}{k = n}}} \Big) \abs{\rho_{ijk,lmn}} + \sum_{i \neq l}\abs{\sum_{j,k}\rho_{ijk,ljk}}\\
    & = \Big(\sum_{\overset{i \neq l}{\overset{j \neq m}{k \neq n}}} + \sum_{{\overset{i = l}{\overset{j \neq m}{k \neq n}}}}\Big) \abs{\rho_{ijk,lmn}}\\
    & \ge 0, 
\end{align}
Above we used the fact that $\abs{\sum_i z_i} \le \sum_i \abs{z_i}$, $\forall \  z_i \in \mathbb{C}$. This completes the proof.
\end{proof}
 Now, following the same reasoning:
 \begin{teo}
If the reduced quantum systems $\rho_A, \rho_B, \rho_C$ satisfy
\begin{align}
    & C_{l_1}(\rho_A) = \sum_{i \neq l}\abs{\sum_{j,k}\rho_{ijk,ljk}} =  \sum_{i \neq l}\sum_{j,k}\abs{\rho_{ijk,ljk}}, \label{eq:ca}\\
    & C_{l_1}(\rho_B) = \sum_{j \neq m}\abs{\sum_{j,k}\rho_{ijk,imk}} =  \sum_{j \neq m}\sum_{i,k}\abs{\rho_{ijk,imk}},\label{eq:cb} \\
    & C_{l_1}(\rho_C) = \sum_{k \neq n}\abs{\sum_{i,j}\rho_{ijk,ijn}} = \sum_{k \neq n}\sum_{i,j}\abs{\rho_{ijk,ijn}}, \label{eq:cc}\\
\end{align}
then 
\begin{equation}
    C^c_{l_1}(\rho_{ABC}) \ge C^c_{l_1}(\rho_{AB}) + C^c_{l_1}(\rho_{AC}) + C^c_{l_1}(\rho_{BC}).
\end{equation}
\end{teo}
 \begin{proof}
 \begin{align}
    C^c_{l_1}(\rho_{ABC})& - C^c_{l_1}(\rho_{AB}) - C^c_{l_1}(\rho_{AC}) - C^c_{l_1}(\rho_{BC})  = C_{l_1}(\rho_{ABC}) + \sum_{\alpha = A,B,C} C_{l_1}(\rho_\alpha) -  \sum_{\alpha < \beta = A,B,C}C_{l_1}(\rho_{\alpha \beta})\\
    & \ge \Big(\sum_{\overset{i \neq l}{\overset{j \neq m}{k \neq n}}} - \sum_{\overset{i \neq l}{\overset{j = m}{k = n}}} - \sum_{\overset{i = l}{\overset{j \neq m}{k = n}}} - \sum_{\overset{i = l}{\overset{j = m}{k \neq n}}} \Big) \abs{\rho_{ijk,lmn}} + \sum_{i \neq l}\abs{\sum_{j,k}\rho_{ijk,ljk}} + \sum_{j \neq m}\abs{\sum_{i,k}\rho_{ijk,imk}} + \sum_{l \neq n}\abs{\sum_{i,j}\rho_{ijk,ijn}}\\
    & = \sum_{\overset{i \neq l}{\overset{j \neq m}{k \neq n}}}\abs{\rho_{ijk,lmn}}\\
    & \ge 0.
\end{align}
\end{proof}
This kind of trade-off relation is equivalent to the monogamy relation expressed by the Eq. (\ref{eq:mon}) for the correlated coherence.
\begin{teo}
A tripartite quantum system that satisfies the trade-off relation
\begin{equation}
    C^c(\rho_{ABC}) \ge C^c(\rho_{AB}) + C^c(\rho_{AC}) + C^c(\rho_{BC})
\end{equation}
also satisfies the monogamy relation
\begin{equation}
    C^c(\rho_{A|BC}) \ge C^c(\rho_{AB}) + C^c(\rho_{AC}),
\end{equation}
for any coherence measure, where $C^c(\rho_{A|BC}) := C^c(\rho_{ABC}) - C(\rho_{A}) - C(\rho_{BC}) $ denotes the correlated coherence between $A$ and $BC$. 
\end{teo}
\begin{proof}
The proof follows directly from the definition
\begin{align}
     C^c(\rho_{A|BC}) - C^c(\rho_{AB}) - C^c(\rho_{AC}) & = C(\rho_{ABC}) + \sum_{\alpha = A,B,C} C(\rho_\alpha) -  \sum_{\alpha < \beta = A,B,C}C(\rho_{\alpha \beta}) \\
     & = C^c(\rho_{ABC}) - C^c(\rho_{AB}) - C^c(\rho_{AC}) - C^c(\rho_{BC})\\
     & \ge 0.
\end{align}
\end{proof}
For instance, the pure quantum state $\ket{G,H,Z,W}_{A,B,C} = \lambda_1 \ket{0,0,0}_{A,B,C} + \lambda_2 \ket{0,0,1}_{A,B,C} + \lambda_3\ket{0,1,0}_{A,B,C} + \lambda_4\ket{1,0,0}_{A,B,C} + \lambda_5 \ket{1,1,1}_{A,B,C}$, with $\abs{\lambda_1}^2 + \abs{\lambda_2}^2 + \abs{\lambda_3}^2 + \abs{\lambda_4}^2 + \abs{\lambda_5}^2= 1$, satisfies the monogamy relation for the correlated coherence. Also, for the state $\ket{\Phi}_{A,B,C} = a_{000} \ket{0,0,0}_{A,B,C} + a_{101}\ket{1,0,1}_{A,B,C} + a_{1,1,0}\ket{1,1,0}_{A,B,C} + a_{111}\ket{1,1,1}_{A,B,C}$ such that $\abs{a_{000}}^2 + \abs{a_{101}}^2 + \abs{a_{110}}^2 + \abs{a_{111}}^2 = 1$, we have $C^c_{l_1}(\rho_{A|BC}) - C^c_{l_1}(\rho_{AB}) - C^c_{l_1}(\rho_{AC}) = 2 \abs{a_{000}a^*_{111}} \ge 0$. This state was considered by Giorgi \cite{Giorgi} in the form $\ket{\Phi(p,\epsilon)} = \sqrt{p \epsilon} \ket{0,0,0}_{A,B,C} + \sqrt{p(1 - \epsilon)}\ket{1,1,1}_{A,B,C} + \sqrt{(1 - p)/2}(\ket{1,1,0}_{A,B,C} + \ket{1,0,1}_{A,B,C})$, with $p, \epsilon \in [0,1]$. Hence, for the correlated coherence, $\ket{\Phi(p,\epsilon)}$ is monogamous for any value of $p, \epsilon \in [0,1]$. In Fig. \ref{fig:mes}, we plotted $M := C^c_{l_1}(\rho_{A|BC}) - C^c_{l_1}(\rho_{AB}) - C^c_{l_1}(\rho_{AC})$ as function of $p$ and $\epsilon$.

\begin{figure}[t]
\includegraphics[scale=0.6]{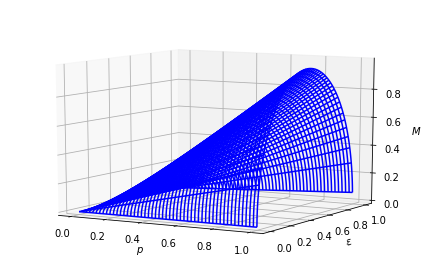}
\caption{The function $M := C^c_{l_1}(\rho_{A|BC}) - C^c_{l_1}(\rho_{AB}) - C^c_{l_1}(\rho_{AC})$ as function of $p$ and $\epsilon$ for the state $\ket{\Phi(p,\epsilon)}$. }
\label{fig:mes}
\end{figure}

Also, let's consider the following state  $\ket{\Psi}_{A,B,C} = \lambda_1 \ket{0,0,0}_{A,B,C} + \lambda_2\ket{0,1,1}_{A,B,C} + \lambda_3 \ket{1,0,0}_{A,B,C} + \lambda_4 \ket{1,1,1}_{A,B,C}$, with $\abs{\lambda_1}^2 + \abs{\lambda_2}^2 + \abs{\lambda_3}^2 + \abs{\lambda_4}^2 = 1$ \cite{Acin}, such that $C_{l_1}(\rho_A) = 2\abs{\lambda_1 \lambda_3^* + \lambda_2 \lambda_4^*} \neq 2(\abs{\lambda_1 \lambda_3^*} + \abs{\lambda_2 \lambda_4^*})$. Then
\begin{align}
    C^c_{l_1}(\rho_{A|BC}) - C^c_{l_1}(\rho_{AB}) - C^c_{l_1}(\rho_{AC}) & = 2( \abs{\lambda_1 \lambda^*_4} + \abs{\lambda_2 \lambda^*_3}) + 2(\abs{\lambda_1 \lambda^*_3 + \lambda_2 \lambda^*_4} - \abs{\lambda_1 \lambda^*_3} - \abs{\lambda_2 \lambda^*_4})\\& + 2(\abs{\lambda_1 \lambda^*_2} + \abs{\lambda_3 \lambda^*_4}- \abs{\lambda_1 \lambda^*_2 + \lambda_3 \lambda^*_4})\nonumber \\
    & \ge 2( \abs{\lambda_1 \lambda^*_4} + \abs{\lambda_2 \lambda^*_3}) + 2(\abs{\lambda_1 \lambda^*_3 + \lambda_2 \lambda^*_4} - \abs{\lambda_1 \lambda^*_3} - \abs{\lambda_2 \lambda^*_4})\\
    & \ge 0.
\end{align}
To prove the last passage, let's suppose, by contradiction, that $\abs{\lambda_1 \lambda^*_3} + \abs{\lambda_2 \lambda^*_4} \ge \abs{\lambda_1 \lambda^*_4} + \abs{\lambda_2 \lambda^*_3} + \abs{\lambda_1 \lambda^*_3 + \lambda_2 \lambda^*_4}$. Squaring this expression, we obtain
\begin{align}
    0 & \ge \abs{\lambda_1 \lambda^*_4}^2 + \abs{\lambda_2 \lambda^*_3}^2 + 2(\abs{\lambda_1 \lambda^*_3} + \abs{\lambda_2 \lambda^*_4})\abs{\lambda_1 \lambda^*_3 + \lambda_2 \lambda^*_4} + 2\mathbb{R}e(\lambda_1 \lambda_3^* \lambda_2 \lambda_4^*)\\
    & = \abs{\lambda_1 \lambda_3^* + \lambda_2 \lambda_4^*}^2 + 2(\abs{\lambda_1 \lambda^*_3} + \abs{\lambda_2 \lambda^*_4})\abs{\lambda_1 \lambda^*_3 + \lambda_2 \lambda^*_4}, 
\end{align}
which is an absurd because the right-hand side is the sum of positive real numbers. On the other hand, if the coefficients $\{\lambda_i\}_{i = 1}^{4}$ are real and positive, the monogamy is obviously satisfied. We can check this considering the state in the form $\ket{\Psi(p,\epsilon)} = \sqrt{p \epsilon} \ket{0,0,0}_{A,B,C} + \sqrt{p(1 - \epsilon)}\ket{1,1,1}_{A,B,C} + \sqrt{(1 - p)/2}(\ket{1,0,0}_{A,B,C} + \ket{0,1,1}_{A,B,C})$, with $p, \epsilon \in [0,1]$, where, in Fig. \ref{fig:mes1}, we plotted $M := C^c_{l_1}(\rho_{A|BC}) - C^c_{l_1}(\rho_{AB}) - C^c_{l_1}(\rho_{AC})$ as function of $p$ and $\epsilon$.
Therefore, the conditions expressed by the equations (\ref{eq:ca}), (\ref{eq:cb}), and (\ref{eq:cc}) are sufficient but not necessary, and seems reasonably to conjecture that the monogamy relation $C^c_{l_1}(\rho_{A|BC}) \ge C^c_{l_1}(\rho_{AB}) + C^c_{l_1}(\rho_{AC})$ holds for any tripartite pure quantum state. Finally, it is possible to establish a weaker trade-off relation for an arbitrary tripartite quantum state, i.e.,
 \begin{align}
     C_{l_1}^c(\rho_{ABC}) \ge \frac{1}{2}\Big(C_{l_1}^c(\rho_{AB}) + C_{l_1}^c(\rho_{AC}) + C_{l_1}^c(\rho_{BC})\Big),
 \end{align}
once
\begin{align}
     C_{l_1}^c(\rho_{ABC}) - \frac{1}{2}\Big(C_{l_1}^c(\rho_{AB}) + C_{l_1}^c(\rho_{AC}) + C_{l_1}^c(\rho_{BC})\Big) & =      C_{l_1}(\rho_{ABC}) - \frac{1}{2}\Big(C_{l_1}(\rho_{AB}) + C_{l_1}(\rho_{AC}) + C_{l_1}(\rho_{BC})\Big)\\
     & \ge \Big(\sum_{\overset{i \neq l}{\overset{j \neq m}{k \neq n}}} + \frac{1}{2}(\sum_{\overset{i = l}{\overset{j \neq m}{k \neq n}}} + \sum_{\overset{i \neq l}{\overset{j = m}{k \neq n}}} + \sum_{\overset{i \neq l}{\overset{j \neq m}{k = n}}})\Big) \abs{\rho_{ijk,lmn}}\\
     & \ge 0.
\end{align}
\begin{figure}[t]
\includegraphics[scale=0.6]{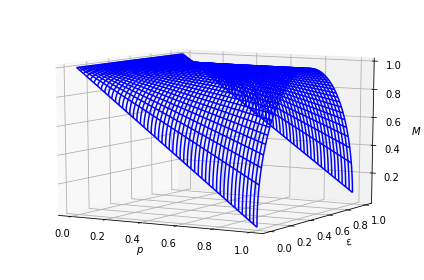}
\caption{The function $M := C^c_{l_1}(\rho_{A|BC}) - C^c_{l_1}(\rho_{AB}) - C^c_{l_1}(\rho_{AC})$ as function of $p$ and $\epsilon$ for the state $\ket{\Psi(p,\epsilon)}$.}
\label{fig:mes1}
\end{figure}

\subsection{Relative entropy of correlated coherence}
First, we'll recall an upper bound for the correlated coherence of a bipartite quantum system $\rho_{AB}$ \cite{Ma}:
 \begin{align}
      C_{re}^c(\rho_{AB})& = C_{re}(\rho_{AB}) - C_{re}(\rho_{A}) - C_{re}(\rho_{B})\\
      & = S(\rho_{{AB}_{diag}}) - S(\rho_{{AB}}) - S(\rho_{{A}_{diag}}) + S(\rho_{A}) - S(\rho_{{B}_{diag}}) + S(\rho_{{B}})\\
      & \le S(\rho_{A}) + S(\rho_{{B}}) - S(\rho_{{AB}}),
 \end{align}
since $S(\rho_{{AB}_{diag}}) -  S(\rho_{{A}_{diag}}) - S(\rho_{{B}_{diag}}) \le 0$, by the subadditivity of von Neumann's entropy \cite{nielsen}. In addition, Xi et al. \cite{Xi} proved that $ C_{re}(\rho_{AB}) \ge C_{re}(\rho_{A}) + C_{re}(\rho_{B})$, by using the fact that local operations never increase the relative entropy, which implies $ C_{re}^c(\rho_{AB}) \ge 0$. Thus \begin{equation}
    0 \le C_{re}^c(\rho_{AB})  \le S(\rho_{A}) + S(\rho_{{B}}) - S(\rho_{{AB}}) := \mathcal{I}_{A,B},
\end{equation}
where $\mathcal{I}_{A,B}$ is the quantum mutual information, which is the total amount of correlations in any bipartite quantum state \cite{Groisman, Schumacher}. Also, it's known that the quantification of quantum coherence depends on a particular reference basis, i.e., these quantum coherence measures are basis-dependent \cite{Byrnes}. Following Ma et al. \cite{Ma}, we'll define the basis-independent relative entropy of coherence, or intrinsic relative entropy of quantum coherence (IREQC). Since the maximally mixed state, $I/d$, is the only basis-independent incoherent state, the IREQC is defined by taken the maximally mixed state as the reference incoherent state, i.e., 
\begin{equation}
C^I_{re}(\rho) := S(\rho||I/d) = \log d - S(\rho),
\end{equation}
where $S(\rho||I/d) = \Tr(\rho \ln \rho - \rho \ln I/d)$ is the relative entropy. Now, defining the intrinsic relative entropy of correlated coherence (IRECC) for a bipartite quantum system as \begin{equation}
C^{Ic}_{re}(\rho_{AB}):= C^{I}_{re}(\rho_{AB}) - C^{I}_{re}(\rho_{A}) - C^{I}_{re}(\rho_{A}),
\end{equation}
we see it is equal to the total amount of correlations in any bipartite quantum system:
\begin{teo}
 The intrinsic relative entropy of correlated coherence of a bipartite quantum system, $C_{re}^{Ic}(\rho_{AB})$, is equal to the quantum mutual information $\mathcal{I}_{A,B} = S(\rho_A) + S(\rho_B) - S(\rho_{AB})$.
\end{teo}
\begin{proof}
The result follows directly from the definition:
\begin{align}
    C_{re}^{Ic}(\rho_{AB}) & = C_{re}^{I}(\rho_{AB}) - C_{re}^{I}(\rho_{A}) - C_{re}^{I}(\rho_{B}) = S(\rho_A) + S(\rho_B) - S(\rho_{AB})\\
    & = \mathcal{I}_{A,B}.
\end{align}
\end{proof}
Now, by considering a tripartite quantum system, the IRECC is given by $C^{Ic}_{re}(\rho_{ABC}) = C^{I}_{re}(\rho_{ABC}) - \sum_{\alpha = A,B,C}C^{I}_{re}(\rho_{\alpha})$. Hence, we have the following trade-off relation
\begin{align}
C^{Ic}_{re}(\rho_{ABC})& - C^{Ic}_{re}(\rho_{AB}) - C^{Ic}_{re}(\rho_{AC}) = C^{I}_{re}(\rho_{ABC}) + C^{I}_{re}(\rho_{A}) - C^{I}_{re}(\rho_{AB}) - C^{Ic}_{re}(\rho_{AB})\\
& = \log(d_A d_B d_C) - S(\rho_{ABC}) + \log d_A - S(\rho_A)  - \log (d_A d_B) +  S(\rho_{AB}) - \log (d_A d_C) +  S(\rho_{AC})\\
& = S(\rho_{AB}) + S(\rho_{AC}) - S(\rho_{ABC}) - S(\rho_A)\\
& \ge 0,
\end{align}
since $S(\rho_{AB}) + S(\rho_{AC}) - S(\rho_{ABC}) - S(\rho_A) \ge 0$ (by the strong subadditivity of von Neumann's entropy \cite{nielsen}). 
\begin{teo}
 If a tripartite quantum system $\rho_{ABC}$ is pure, then $\rho_{ABC}$ satisfies the trade-off relation
\begin{align}
     C^{Ic}_{re}(\rho_{ABC}) = C^{Ic}_{re}(\rho_{AB}) + C^{Ic}_{re}(\rho_{AC}) + C^{Ic}_{re}(\rho_{BC}),
 \end{align}
 and consequently the monogamy relation
 \begin{equation}
     C^{Ic}_{re}(\rho_{A|BC}) = C^{Ic}_{re}(\rho_{AB}) + C^{Ic}_{re}(\rho_{AC}).
 \end{equation}
\end{teo}
\begin{proof}
\begin{align}
C^{Ic}_{re}(\rho_{ABC})- \sum_{\alpha < \beta = A,B,C}C^{Ic}_{re}(\rho_{\alpha\beta}) & = C^{I}_{re}(\rho_{ABC}) + \sum_{\alpha = A,B,C}C^{I}_{re}(\rho_{\alpha}) - \sum_{\alpha < \beta = A,B,C}C^{I}_{re}(\rho_{\alpha\beta})\\
& = S(\rho_{AB}) + S(\rho_{AC}) + S(\rho_{BC}) - S(\rho_{ABC}) - S(\rho_{A}) - S(\rho_{B}) - S(\rho_{C})\\
& = \mathcal{T_{A,B,C}},
\end{align}
since $\mathcal{T_{A,B,C}}$ is the interaction information \cite{Cover}, and $\mathcal{T_{A,B,C}} = 0$ for tripartite pure states, since $S(\rho_{ABC}) = 0$, and $S(\rho_{AB}) = S(\rho_{C})$, $S(\rho_{AC}) = S(\rho_{B})$, $S(\rho_{BC}) = S(\rho_{A})$ \cite{Witten}. This completes the proof.
\end{proof}
Consequently, $C^{Ic}_{re}(\rho_{A|BC}) = C^{Ic}_{re}(\rho_{AB}) + C^{Ic}_{re}(\rho_{AC})$ is equivalent to $\mathcal{I}_{A|BC} = \mathcal{I}_{A,B} + \mathcal{I}_{A,B}$ \cite{Costa}, once that $C^{Ic}_{re}(\rho_{A|BC}) = \mathcal{I}_{A|BC}$, $C^{Ic}_{re}(\rho_{XY}) = \mathcal{I}_{X,Y}$, $\forall \ X,Y = A,B,C$ such that $X \neq Y$.

\section{Conclusions}
\label{sec:con}
Monogamy relations are an important feature of quantum correlations, as they tells us that a specific quantum resource cannot be shared freely. Recently, the correlated coherence was used as a resource for remote state preparation and quantum teleportation \cite{Li}. Hence, it is important to know if the correlated coherence satisfies monogamy relations. In this paper, we have studied the monogamy properties of the correlated coherence for the $l_1$-norm and relative entropy measures. For the $l_1$-norm, the correlated coherence is monogamous for a given class of quantum states, and we conjectured that it is monogamous, at least, for tripartite pure quantum states. For the relative entropy of coherence, and using a maximally mixed state as the reference incoherent state, we showed that the correlated coherence is monogamous for tripartite pure quantum system. Another interesting finding is that the intrinsic relative entropy of correlated coherence (IRECC) is equal to quantum mutual information. Finally, we also established some trade-off relations between tripartite and bipartite quantum systems, and proved that the trade-off relation $C^c(\rho_{ABC}) \ge C^c(\rho_{AB}) + C^c(\rho_{AC}) + C^c(\rho_{BC})$ is equivalent to the monogamy relation $C^c(\rho_{A|BC}) \ge C^c(\rho_{AB}) + C^c(\rho_{AC})$.

\begin{acknowledgments}
This work was supported by the Coordena\c{c}\~ao de Aperfei\c{c}oamento de Pessoal de N\'ivel Superior (CAPES), process 88882.427924/2019-01, and by the Instituto Nacional de Ci\^encia e Tecnologia de Informa\c{c}\~ao Qu\^antica (INCT-IQ), process 465469/2014-0.
\end{acknowledgments}

\end{document}